\documentclass[pdftex]{llncs}

\usepackage{etex}

\usepackage[utf8]{inputenc}

\usepackage[T1]{fontenc}
\usepackage{textcomp}

\usepackage{microtype}
\usepackage[british]{babel}
\usepackage{soul}
\usepackage[colorlinks,citecolor=blue,linkcolor=red]{hyperref}

\usepackage{mdframed}
\usepackage{wrapfig}

\usepackage{amssymb}
\usepackage{amsmath}

\usepackage{tabu,booktabs}
\usepackage{multicol,multirow}
\usepackage[inline,shortlabels]{enumitem}

\usepackage{tikz}
\usetikzlibrary{calc, patterns, shapes}

\usepackage[noline,noend]{algorithm2e}
\AlgoDontDisplayBlockMarkers
\DontPrintSemicolon
\LinesNumbered
\let\oldnl\nl
\newcommand{\nonl}{\renewcommand{\nl}{\let\nl\oldnl}}
\SetAlgoCaptionLayout{normal}
\SetAlgoCaptionSeparator{}
\SetAlCapSkip{.5em}
\makeatletter
\renewcommand{\algocf@makecaption}[2]{%
	\addtolength{\hsize}{1.5\algomargin}%
	\parbox[t]{\hsize}{\algocf@captiontext{#1:}{#2}}%
	\addtolength{\hsize}{-1.5\algomargin}%
}
\makeatother
\SetKwProg{Def}{def}{:}{}
\SetKwIF{If}{ElseIf}{Else}{if}{:}{elif}{else:}{}
\SetKwFor{For}{for}{:}{}
\SetKwFor{While}{while}{:}{}
\SetKwIF{Catch}{}{Try}{catch}{:}{}{try:}{}
\SetKwFor{Loop}{loop:}{}{}
\SetKwFor{DoPar}{do in parallel:}{}{}
\SetArgSty{textup}
\SetFuncArgSty{textrm}
\SetKw{Break}{break}
\SetKw{Continue}{continue}
\SetKw{Raise}{raise}
\SetKw{False}{false}
\SetKw{True}{true}
\SetKwFunction{Solve}{solve}
\SetKwFunction{Search}{search}
\SetKwFunction{TangleAttr}{TAttr}

\newcommand{\Attr}{\mathit{Attr}}
\newcommand{\TAttr}{\mathit{TAttr}}

\newcommand{\pr}{\textsf{pr}}
\renewcommand{\r}{\textsf{r}}
\newcommand{\dom}{\textsf{dom}}
\newcommand{\invpr}{\pr^{-1}}
\newcommand{\invalpha}{{\overline{\alpha}}}

\newcommand{\sqdiamond}{\tikz [x=1.2ex,y=1.2ex,line width=.08ex] \draw (0,.5) -- (.5,1) -- (1,.5) -- (.5,0) -- (0,.5) -- cycle;}
\newcommand{\sqsq}{\tikz [x=0.95ex,y=1ex,line width=.1ex] \draw (0,0) -- (1,0) -- (1,1) -- (0,1) -- (0,0) -- cycle;}

\newcommand{\Even}{{\scalebox{0.95}{\sqdiamond}}}
\newcommand{\Odd}{{\scalebox{0.9}{$\sqsq$}}}

\newcommand{\Veven}{V_{\Even}}
\newcommand{\Vodd}{V_{\Odd}}

\title{Attracting Tangles to Solve Parity Games}

\author{
  Tom van Dijk\thanks{The author is supported by the FWF, NFN Grant S11408-N23 (RiSE)}
}

\institute{
  Formal Models and Verification \\ Johannes Kepler University, Linz \\
  \email{tom.vandijk@jku.at}
}

\begin{document}
\maketitle
\begin{abstract}
Parity games have important practical applications in formal verification and synthesis,
especially to solve the model-checking problem of the modal mu-calculus.
They are also interesting from the theory perspective,
because they are widely believed to admit a polynomial solution,
but so far no such algorithm is known.

We propose a new algorithm to solve parity games based on learning tangles, which
are strongly connected subgraphs for which one player has a strategy to win all cycles in the subgraph.
We argue that tangles play a fundamental role in the prominent parity game solving algorithms.
We show that tangle learning is competitive in practice and the fastest solver for large random games.

\end{abstract}

\section{Introduction}
\label{sec:introduction}

Parity games are turn-based games played on a finite graph.
Two players \emph{Odd} and \emph{Even} play an infinite game by moving a token along the edges of the graph.
Each vertex is labeled with a natural number \emph{priority} and the winner of the game is determined by the parity of the highest priority that is encountered infinitely often.
Player Odd wins if this parity is odd; otherwise, player Even wins.

Parity games are interesting both for their practical applications and for complexity theoretic reasons.
Their study has been motivated by their relation to many problems in formal verification and synthesis that can be reduced to the problem of solving parity games, as parity games capture the expressive power of nested least and greatest fixpoint operators~\cite{DBLP:conf/cav/Fearnley17}.
In particular, deciding the winner of a parity game is polynomial-time equivalent to checking non-emptiness of non-deterministic parity tree automata~\cite{DBLP:conf/stoc/KupfermanV98}, and to the explicit model-checking problem of the modal $\mu$-calculus~\cite{DBLP:journals/tcs/EmersonJS01,DBLP:conf/dagstuhl/2001automata,DBLP:journals/tcs/Kozen83}.

Parity games are interesting in complexity theory, 
as the problem of determining the winner of a parity game is known to lie in $\text{UP}\cap\text{co-UP}$~\cite{DBLP:journals/ipl/Jurdzinski98},
%
which is contained in $\text{NP}\cap\text{co-NP}$~\cite{DBLP:journals/tcs/EmersonJS01}. 
This problem is therefore unlikely to be NP-complete and it is widely believed that a polynomial solution exists.
Despite much effort, such an algorithm has not been found yet.

The main contribution of this paper is based on the notion of a \emph{tangle}.
A tangle is a strongly connected subgraph of a parity game for which one of the players has a strategy to win all cycles in the subgraph.
We propose this notion and its relation to dominions and cycles in a parity game.
Tangles are related to snares~\cite{DBLP:conf/lpar/Fearnley10} and quasi-dominions~\cite{DBLP:conf/cav/BenerecettiDM16}, with the critical difference that tangles are strongly connected, whereas snares and quasi-dominions may be unconnected as well as contain vertices that are not in any cycles.
We argue that tangles play a fundamental role in various parity game algorithms, in particular in priority promotion~\cite{DBLP:conf/cav/BenerecettiDM16,Benerecetti2018}, Zielonka's recursive algorithm~\cite{DBLP:journals/tcs/Zielonka98}, strategy improvement~\cite{DBLP:conf/lpar/Fearnley10,DBLP:conf/cav/Fearnley17,DBLP:conf/cav/VogeJ00}, small progress measures~\cite{DBLP:conf/stacs/Jurdzinski00}, and in the recently proposed quasi-polynomial time progress measures~\cite{DBLP:conf/stoc/CaludeJKL017,DBLP:conf/spin/FearnleyJS0W17}.

The core insight of this paper is that tangles can be used to attract sets of vertices at once, since the losing player is forced to escape a tangle.
This leads to a novel algorithm to solve parity games called \emph{tangle learning}, which is based on searching for tangles along a top-down $\alpha$-maximal decomposition of the parity game.
New tangles are then attracted in the next decomposition.
This naturally leads to learning nested tangles and, eventually, finding dominions.
We prove that tangle learning solves parity games and present several extensions to the core algorithm,
including \emph{alternating} tangle learning, where the two players take turns maximally searching for tangles in their regions, and \emph{on-the-fly} tangle learning, where newly learned tangles immediately refine the decomposition.


We relate the complexity of tangle learning to the number of learned tangles before finding a dominion, which is related to how often the solver is distracted by paths to higher winning priorities that are not suitable strategies.

We evaluate tangle learning in a comparison based on the parity game solver Oink~\cite{Oink2018},
using the benchmarks of Keiren~\cite{DBLP:conf/fsen/Keiren15} as well as random parity games of various sizes.
We compare tangle learning to priority promotion~\cite{DBLP:conf/cav/BenerecettiDM16,Benerecetti2018} and to Zielonka's recursive algorithm~\cite{DBLP:journals/tcs/Zielonka98} as implemented in Oink.

\section{Preliminaries}
\label{sec:preliminaries}

Parity games are two-player turn-based infinite-duration games over a finite directed graph $G=(V,E)$, where
every vertex belongs to exactly one of two players called player \emph{Even} and player \emph{Odd}, 
and where every vertex is assigned a natural number called the \emph{priority}.
Starting from some initial vertex, a play of both players is an infinite path in $G$ where the owner of each vertex determines the next move. The winner of such an infinite play is determined by the parity of the highest priority that occurs infinitely often along the play.

More formally, a parity game $\Game$ is a tuple $(\Veven, \Vodd, E, \pr)$ where $V=\Veven\cup \Vodd$ is a set of vertices partitioned into the sets $\Veven$ controlled by player \emph{Even} and $\Vodd$ controlled by player \emph{Odd}, and $E\subseteq V\times V$ is a left-total binary relation describing all moves, i.e., every vertex has at least one successor.
We also write $E(u)$ for all successors of $u$ and $u\rightarrow v$ for $v\in E(u)$.
The function $\pr\colon V\rightarrow \{0,1,\dotsc,d\}$
assigns to each vertex a \emph{priority}, where $d$ is the highest priority in the game.



We write $\pr(v)$ for the priority of a vertex $v$
and $\pr(V)$ for the highest priority of vertices $V$ and
$\pr(\Game)$ for the highest priority in the game $\Game$.
Furthermore, we write $\pr^{-1}(i)$ for all vertices with the priority $i$.
A \emph{path} $\pi=v_0 v_1 \dots$ is a sequence of vertices consistent with $E$, i.e.,
$v_i \rightarrow v_{i+1}$ for all successive vertices.
A \emph{play} is an infinite path.
We denote with $\inf(\pi)$ the vertices in $\pi$ that occur infinitely many times in $\pi$.
Player Even wins a play $\pi$ if $\pr(\inf(\pi))$ is even; player Odd wins if $\pr(\inf(\pi))$ is odd.
We write $\text{Plays}(v)$ to denote all plays starting at vertex $v$.
%

A \emph{strategy} $\sigma\colon V\to V$ is a partial function that assigns to each vertex in its domain a single successor in $E$, i.e., $\sigma\subseteq E$.
We refer to a strategy of player $\alpha$ to restrict the domain of $\sigma$ to $V_\alpha$. 
In the remainder, all strategies $\sigma$ are of a player $\alpha$.
We write $\text{Plays}(v, \sigma)$ for the set of plays from $v$ consistent with $\sigma$, and $\text{Plays}(V,\sigma)$ for $\{\,\pi\in\text{Plays}(v,\sigma)\mid v\in V\,\}$.

A fundamental result for parity games is that they are memoryless determined~\cite{DBLP:conf/focs/EmersonJ91}, i.e., each vertex is either winning for player Even or for player Odd, and both players have a strategy for their winning vertices.
Player $\alpha$ wins vertex $v$ if they have a strategy $\sigma$ such that all plays in $\text{Plays}(v,\sigma)$ are winning for player $\alpha$.

%




Several algorithms for solving parity games employ \emph{attractor computation}.
Given a set of vertices $A$,
the attractor of $A$ for a player $\alpha$ represents those vertices from which player $\alpha$ can force a play to visit $A$.
%
We write $\Attr^\Game_\alpha(A)$ to attract vertices in $\Game$ to $A$ as player $\alpha$,
i.e.,
\[\mu Z \,.\, A \cup \{\;v\in V_\alpha \mid E(v)\cap Z \neq \emptyset\;\} \cup \{\;v\in V_{\invalpha} \mid E(v)\subseteq Z\;\}\]
%
Informally, we compute the $\alpha$-attractor of $A$ with a backward search from $A$, initially setting $Z:=A$ and iteratively adding $\alpha$-vertices with a successor in $Z$ and $\invalpha$-vertices with no successors outside $Z$.
We also obtain a strategy $\sigma$ for player $\alpha$, starting with an empty strategy, by selecting a successor in $Z$ when we attract vertices of player $\alpha$ and when the backward search finds a successor in $Z$ for the $\alpha$-vertices in $A$.
We call a set of vertices $A$ $\alpha$-maximal if $A=\Attr^\Game_\alpha(A)$. 

A \emph{dominion} $D$ is a set of vertices for which player $\alpha$ has a strategy $\sigma$ such that all plays consistent with $\sigma$ stay in $D$ and are winning for player $\alpha$.
We also write a \emph{$p$-dominion} for a dominion where $p$ is the highest priority encountered infinitely often in plays consistent with $\sigma$, i.e., $p:=\max\{\,\pr(\inf(\pi))\mid \pi\in\textrm{Plays}(D,\sigma)\,\}$.

\section{Tangles}
\label{sec:tangles}




\begin{definition}

A \emph{$p$-tangle} is a nonempty set of vertices $U\subseteq V$ with $p=\pr(U)$, for which player $\alpha \equiv_2 p$ has a strategy $\sigma\colon U_\alpha\rightarrow U$, such that the graph $(U,E')$, with 
$E':= E\cap\big(\sigma\cup(U_\invalpha\times U)\big)$,
is strongly connected
and player $\alpha$ wins all cycles in $(U,E')$.
\end{definition}


Informally, a tangle is a set of vertices for which player $\alpha$ has a strategy to win all cycles inside the tangle.
Thus, player $\invalpha$ loses all plays that stay in $U$ and is therefore forced to escape the tangle.
The highest priority by which player $\alpha$ wins a play in $(U,E')$ is $p$.
We make several basic observations related to tangles.
\begin{enumerate}
	\item {A $p$-tangle from which player $\invalpha$ cannot leave is a $p$-dominion.}
	\item {Every $p$-dominion contains one or more $p$-tangles.} 
	\item {Tangles may contain tangles of a lower priority.}
\end{enumerate}


Observation~1 follows by definition.
Observation~2 follows from the fact that dominions won by player $\alpha$ with some strategy $\sigma$ must contain strongly connected subgraphs where all cycles are won by player $\alpha$ and the highest winning priority is $p$.
For observation~3, consider a $p$-tangle for which player $\invalpha$ has a strategy that avoids priority $p$ while staying in the tangle.
Then there is a $p'$-tangle with $p'<p$ in which player $\invalpha$ also loses.

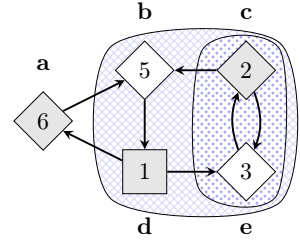
\begin{wrapfigure}{r}{0.4\textwidth}
	\vspace{-2.21em} 
	\centering
	\scalebox{0.9}{
		\begin{tikzpicture}		
		\tikzset{every edge/.append style={>=stealth,->,solid,thick,draw,text height=0.5ex,text depth=0.2ex}}
		\tikzset{every node/.append style={minimum size=6mm,draw,fill=black!10}}
		\tikzset{my label/.style args={#1:#2}{
				append after command={
					($(\tikzlastnode.center)$) coordinate [label={[label distance=5mm,black]#1:\textbf{\strut #2}}]
		}}}
		\tikzstyle{even}=[diamond,minimum size=6mm,draw]
		\tikzstyle{odd}=[regular polygon,regular polygon sides=4,minimum size=6mm,draw]
		\tikzstyle{seven}=[even,fill=white]
		\tikzstyle{sodd}=[odd,fill=white]
		
		\draw[pattern=crosshatch,pattern color=blue!15] plot [smooth cycle, tension=0.5]
		coordinates {(-0.5,1.85) (-0.5,-0.4) (2,-0.45) (2,1.9)};
		\draw[fill,white] plot [smooth cycle, tension=0.5]
		coordinates {(0.85,1.8) (0.8,-0.25) (1.95,-0.35) (2,1.8)};		
		\draw[pattern=crosshatch dots,pattern color=blue!35] plot [smooth cycle, tension=0.5]
		coordinates {(0.85,1.8) (0.8,-0.25) (1.95,-0.35) (2,1.8)};		
		\draw (0,1.5)     node[seven, my label={above:b}] (a) {5};
		\draw (-1.5,0.75) node[even, my label={above:a}] (b) {6};
		\draw (0,0)       node[odd, my label={below:d}]  (c) {1};
		\draw (1.5,0)     node[seven, my label={below:e}] (d) {3};
		\draw (1.5,1.5)   node[even, my label={above:c}] (e) {2};
		\draw (a) edge (c);
		\draw (b) edge (a);
		\draw (c) edge (d);
		\draw (c) edge (b);
		\draw (d) edge [bend left=20] (e);
		\draw (e) edge [bend left=20] (d);
		\draw (e) edge (a);
		\end{tikzpicture}
	}
	\caption{A $5$-dominion with a $5$-tangle and a $3$-tangle}
	\label{fig:tangled}
	\vspace{-2em} 
\end{wrapfigure}

We can in fact find a hierarchy of tangles in any dominion $D$ with winning strategy $\sigma$
by computing the set of winning priorities
$\{\,\pr(\inf(\pi))\mid \pi\in\textrm{Plays}(D,\sigma)\,\}$.
There is a $p$-tangle in $D$ for every $p$ in this set.
Tangles are thus a natural substructure of dominions.

See for example Figure~\ref{fig:tangled}. Player Odd wins this dominion with highest priority $5$ and strategy $\{\,\textbf{d}\rightarrow\textbf{e}\,\}$.
Player Even can also avoid priority $5$ and then loses with priority $3$.
The 5-dominion $\{\, \textbf{a}, \textbf{b}, \textbf{c}, \textbf{d}, \textbf{e} \,\}$ contains the $5$-tangle $\{\, \textbf{b}, \textbf{c}, \textbf{d}, \textbf{e} \,\}$ and the $3$-tangle $\{\, \textbf{c}, \textbf{e} \,\}$.
%
%


\section{Solving by Learning Tangles}
\label{sec:tanglelearning}

Since player $\invalpha$ must escape tangles won by player $\alpha$, we can treat a tangle as an abstract vertex controlled by player $\invalpha$ that can be attracted by player $\alpha$, thus attracting all vertices of the tangle.
This section proposes the \emph{tangle learning} algorithm, which searches for tangles along a top-down $\alpha$-maximal decomposition of the game.
We extend the attractor to attract all vertices in a tangle when player $\invalpha$ is forced to play from the tangle to the attracting set.
After extracting new tangles from regions in the decomposition, we iteratively repeat the procedure until a dominion is found.
We show that tangle learning solves parity games. 

%

\subsection{Attracting tangles}

Given a tangle $t$, we denote its vertices simply by $t$ and its witness strategy by $\sigma_T(t)$.
We write $E_T(t)$ for the edges from $\invalpha$-vertices in the tangle to the rest of the game: $E_T(t) := \{\, v \mid u \to v \land u\in t\cap V_\invalpha \land v \in V\setminus t \,\}$.
We write $T_\Odd$ for all tangles where $\pr(t)$ is odd (won by player Odd) and $T_\Even$ for all tangles where $\pr(t)$ is even.
%
%
We write $\TAttr^{\Game, T}_\alpha(A)$ to attract vertices in $\Game$ and vertices of tangles in $T$ to $A$ as player $\alpha$,
i.e.,
\begin{align*}
\mu Z \,.\, A & \cup \{\;v\in V_\alpha \mid E(v)\cap Z \neq \emptyset\;\} \cup \{\;v\in V_{\invalpha} \mid E(v)\subseteq Z\;\} \\
& \cup \{\; v\in t\mid t\in T_\alpha \land 
E_T(t)\neq\emptyset \land 
E_T(t)\subseteq Z \;\}
\end{align*}

This approach is not the same as the subset construction.
Indeed, we do not add the tangle itself but rather add all its vertices together.
Notice that this attractor does not guarantee arrival in $A$, as player $\invalpha$ can stay in the added tangle, but then player $\invalpha$ loses.

To compute a witness strategy $\sigma$ for player $\alpha$, as with $\Attr^\Game_\alpha$, we select a successor in $Z$ when attracting single vertices of player $\alpha$ and when we find a successor in $Z$ for the $\alpha$-vertices in $A$.
When we attract vertices of a tangle, we update $\sigma$ for each tangle $t$ sequentially, by updating $\sigma$ with the strategy in $\sigma_T(t)$ of those $\alpha$-vertices in the tangle for which we do not yet have a strategy in $\sigma$, i.e., $\{\, (u,v)\in\sigma_T(t)\mid u\notin \textsf{dom}(\sigma) \,\}$.
This is important since tangles can overlap.

In the following, we call a set of vertices $A$ $\alpha$-maximal 
if $A = \TAttr^{\Game, T}_\alpha(A)$.
Given a game $\Game$ and a set of vertices $U$, we write $\Game\cap U$ for the subgame $\Game'$ where $V':=V\cap U$ and $E':=E\cap(V'\times V')$.
Given a set of tangles $T$ and a set of vertices $U$, we write $T\cap U$ for all tangles with all vertices in $U$, i.e., $\{\,t\in T\mid t\subseteq U\,\}$,
and we extend this notation to $T\cap\Game'$ for the tangles in the game $\Game'$, i.e., $T\cap V'$.



\subsection{The \texttt{solve} algorithm}

\begin{algorithm}[t]
	\Def{\Solve{$\Game$}}{
	$W_\Even$ $\leftarrow$ $\emptyset$, $W_\Odd$ $\leftarrow$ $\emptyset$, $\sigma_\Even$ $\leftarrow$ $\emptyset$, $\sigma_\Odd$ $\leftarrow$ $\emptyset$, $T$ $\leftarrow$ $\emptyset$ \;
	\While{$\Game\neq\emptyset$}{
		$T, d$ $\leftarrow$ \Search{$\Game$, $T$} \;
		$\alpha$ $\leftarrow$ $\pr(d) \bmod 2$ \;
		$D,\sigma$ $\leftarrow$ $\Attr_\alpha^\Game(d)$ \;
		$W_\alpha$ $\leftarrow$ $W_\alpha\cup D$, $\sigma_\alpha$ $\leftarrow$ $\sigma_\alpha\cup\sigma_T(d)\cup\sigma$ \;
		$\Game$ $\leftarrow$ $\Game\setminus D$ ,
		$T \leftarrow T \cap (\Game\setminus D)$ \;
	}
	\Return $W_\Even, W_\Odd, \sigma_\Even, \sigma_\Odd$
}
\caption{The \texttt{solve} algorithm which computes the winning regions and winning strategies for both players of a given parity game.}
\label{alg:solve}
\end{algorithm}

We solve parity games by iteratively searching and removing a dominion of the game, as in~\cite{DBLP:conf/cav/BenerecettiDM16,DBLP:journals/siamcomp/JurdzinskiPZ08,DBLP:journals/jcss/Schewe17}.
See Algorithm~\ref{alg:solve}.
The \texttt{search} algorithm (described below) is given a game and a set of tangles and returns an updated set of tangles and a tangle $d$ that is a dominion.
Since the dominion $d$ is a tangle, we derive the winner $\alpha$ from the highest priority (line~5) and use standard attractor computation to compute a dominion $D$ (line~6).
We add the dominion to the winning region of player $\alpha$ (line~7).
We also update the winning strategy of player $\alpha$ using the witness strategy of the tangle $d$ plus the strategy $\sigma$ obtained during attractor computation.
To solve the remainder, we remove all solved vertices from the game and we remove all tangles that contain solved vertices (line~8).
When the entire game is solved, we return the winning regions and winning strategies of both players (lines~9).
Reusing the (pruned) set of tangles for the next \texttt{search} call is optional; if \texttt{search} is always called with an empty set of tangles, the ``forgotten'' tangles would be found again.

\subsection{The \texttt{search} algorithm}

\begin{algorithm}[t]
	\Def{\Search{$\Game$, $T$}}{
		\While{$\True$}{
		$\r$ $\leftarrow$ $\emptyset$, $Y$ $\leftarrow$ $\emptyset$ \;
		\While{$\Game\setminus\r \neq \emptyset$}{
			$\Game'$ $\leftarrow$ $\Game \setminus \r$, $T'\leftarrow T\cap(\Game\setminus\r)$ \;
			$p$ $\leftarrow$ $\pr(\Game')$, $\alpha$ $\leftarrow$ $\pr(\Game') \bmod 2$ \;
			$Z,\sigma$ $\leftarrow$ $\TAttr^{\Game',T'}_\alpha\big(\{\, v \in \Game' \mid \pr(v)=p \,\}\big)$ \; 
			$A$ $\leftarrow$ \texttt{extract-tangles}($Z$, $\sigma$) \;
			\lIf{$\exists\, t\in A \colon E_T(t) = \emptyset$}{
				\Return $T\cup Y$, $t$
			}
			$\textsf{r}$ $\leftarrow$ $\textsf{r} \cup \big(Z\mapsto p \big)$, $Y$ $\leftarrow$ $Y\cup A$ \;
		}
		$T$ $\leftarrow$ $T\cup Y$ \;			
		}
	} 
	\caption{The \texttt{search} algorithm which, given a game and a set of tangles, returns the updated set of tangles and a tangle that is a dominion.}
	\label{alg:search}
\end{algorithm}


The \texttt{search} algorithm is given in Algorithm~\ref{alg:search}. 
The algorithm iteratively computes a top-down decomposition of $\Game$ into sets of vertices called \emph{regions} such that each region is $\alpha$-maximal for the player $\alpha$ who wins the highest priority in the region.
Each next region in the remaining subgame $\Game'$ is obtained by taking all vertices with the highest priority $p$ in $\Game'$ and computing the tangle attractor set of these vertices for the player that wins that priority, i.e., player $\alpha\equiv_2 p$.
%
%
%
%
%
As every next region has a lower priority,
each region is associated with a unique priority $p$.
We record the current region of each vertex in an auxiliary partial function $\textsf{r}\colon V\to \{0,1,\dotsc,d\}$ called the region function.
We record the new tangles found during each decomposition in the set $Y$.


In each iteration of the decomposition, we first obtain the current subgame $\Game'$ (line~5) and the top priority $p$ in $\Game'$ (line~6).
%
We compute the next region by attracting (with tangles) to the vertices of priority $p$ in $\Game'$ (line~7).
We use the procedure \texttt{extract-tangles} (described below) to obtain new tangles from the computed region (line~8).
For each new tangle, we check if the set of outgoing edges to the full game $E_T(t)$ is empty.
If $E_T(t)$ is empty, then we have a dominion and we terminate the procedure (line~9).
If no dominions are found, then we add the new tangles to $Y$ and update $\r$ (line~10).
After fully decomposing the game into regions, we add all new tangles to $T$ (line~11) and restart the procedure.

\subsection{Extracting tangles from a region}
\label{sec:extract_tangles}

To search for tangles in a given region $A$ of player $\alpha$ with strategy $\sigma$, we first remove all vertices where player $\invalpha$ can play to lower regions (in $\Game'$) while player $\alpha$ is constrained to $\sigma$, i.e.,
%
\[\nu Z \,.\, A \cap \big(\{\;v\in V_\invalpha \mid E'(v)\subseteq Z\;\} \cup \{\;v\in V_{\alpha} \mid \sigma(v)\in Z\;\}\big)\]

This procedure can be implemented efficiently with a backward search, starting from all vertices of priority $p$ that escape to lower regions.
Since there can be multiple vertices of priority $p$, the reduced region may consist of multiple unconnected tangles.
We compute all nontrivial bottom SCCs of the reduced region, restricted by the strategy $\sigma$.
Every such SCC is a unique $p$-tangle.

%



\subsection{Tangle learning solves parity games}

We now prove properties of the proposed algorithm.

\begin{lemma}
	\label{lemma:alpha_maximal}
	All regions recorded in $\r$ in Algorithm~\ref{alg:search} are $\alpha$-maximal in their subgame.
\end{lemma}
\begin{proof}
	This is vacuously true at the beginning of the search.
	Every region $Z$ is $\alpha$-maximal as $Z$ is computed with $\TAttr$ (line~7).
	Therefore the lemma remains true when $\r$ is updated at line~10.
	New tangles are only added to $T$ at line~11, after which $\r$ is reset to $\emptyset$.
	\qed
\end{proof}

\begin{lemma}
	\label{lemma:alpha_wins}
	All plays consistent with $\sigma$ that stay in a region  
	are won by player $\alpha$.
\end{lemma}
\begin{proof}
	Based on how the attractor computes the region, we show that all cycles (consistent with $\sigma$) in the region are won by player $\alpha$.
	Initially, $Z$ only contains vertices with priority $p$; therefore, any cycles in $Z$ are won by player $\alpha$.
	We consider two cases:
	\begin{enumerate*}[a)]
	\item When attracting a single vertex $v$, any new cycles must involve vertices with priority $p$ from the initial set $A$, since all other $\alpha$-vertices in $Z$ already have a strategy in $Z$ and all other $\invalpha$-vertices in $Z$ have only successors in $Z$, otherwise they would not be attracted to $Z$.
	Since $p$ is the highest priority in the region, every new cycle is won by player $\alpha$.
	\item When attracting vertices of a tangle, we set the strategy for all attracted vertices of player $\alpha$ to the witness strategy of the tangle.
	Any new cycles either involve vertices with priority $p$ (as above) or are cycles inside the tangle that are won by player $\alpha$.
	\end{enumerate*}
	\qed
\end{proof}

\begin{lemma}
	\label{lemma:play_to_p}
	Player $\invalpha$ can reach a vertex with the highest priority $p$ from every vertex in the region, via a path in the region that is consistent with strategy $\sigma$.
\end{lemma}
\begin{proof}
	The proof is based on how the attractor computes the region.
	This property is trivially true for the initial set of vertices with priority $p$.
	We consider again two cases:
	\begin{enumerate*}[a)]
	\item When attracting a single vertex $v$, $v$ is either an $\alpha$-vertex with a strategy to play to $Z$, or an $\invalpha$-vertex whose successors are all in $Z$. Therefore, the property holds for $v$.
	\item Tangles are strongly connected w.r.t. their witness strategy.
	Therefore player $\invalpha$ can reach every vertex of the tangle and since the tangle is attracted to $Z$, at least one $\invalpha$-vertex can play to $Z$.
	Therefore, the property holds for all attracted vertices of a tangle.
	\end{enumerate*}
	\qed
\end{proof}

\begin{lemma}
	\label{lemma:E_t_higher}
	For each new tangle $t$, all successors of $t$ are in higher $\alpha$-regions.
\end{lemma}
\begin{proof}
	For every bottom SCC $B$ (computed in \texttt{extract-tangles}), $E'(v)\subseteq B$ for all $\invalpha$-vertices $v\in B$, otherwise player $\invalpha$ could leave $B$ and $B$ would not be a bottom SCC.
	Recall that $E'(v)$ is restricted to edges in the subgame $\Game'=\Game\setminus\r$.
	Therefore $E(v)\subseteq\dom(\r)\cup B$ in the full game for all $\invalpha$-vertices $v\in B$.
	Recall that $E_T(t)$ for a tangle $t$ refers to all successors for player $\invalpha$ that leave the tangle.
	Hence, 
	$E_T(t)\subseteq\dom(\r)$ for every tangle $t:=B$.
	Due to Lemma~\ref{lemma:alpha_maximal}, no $\invalpha$-vertex in $B$ can escape to a higher $\invalpha$-region.
	Thus $E_T(t)$ only contains vertices from higher $\alpha$-regions when the new tangle is found by \texttt{extract-tangles}.
	\qed
\end{proof}

\begin{lemma}
	\label{lemma:unique_tangle}
	Every nontrivial bottom SCC $B$ of the reduced region restricted by witness strategy $\sigma$ is a unique $p$-tangle.
\end{lemma}
\begin{proof}
	All $\alpha$-vertices $v$ in $B$ have a strategy $\sigma(v)\in B$, since $B$ is a bottom SCC when restricted by $\sigma$.
	$B$ is strongly connected by definition.
	Per Lemma~\ref{lemma:alpha_wins}, player $\alpha$ wins all plays consistent with $\sigma$ in the region and therefore also in $B$.
	Thus, $B$ is a tangle.
	Per Lemma~\ref{lemma:play_to_p}, player $\invalpha$ can always reach a vertex of priority $p$, therefore any bottom SCC must include a vertex of priority $p$.
	Since $p$ is the highest priority in the subgame, $B$ is a $p$-tangle.
	Furthermore, the tangle must be unique.
	If the tangle was found before, then per Lemma~\ref{lemma:alpha_maximal} and Lemma~\ref{lemma:E_t_higher},
	it would have been attracted to a higher $\alpha$-region.
	\qed
\end{proof}

\begin{lemma}
	\label{lemma:lowest_tangle}
	The lowest region in the decomposition always contains a tangle.
\end{lemma}
\begin{proof}
	The lowest region is always nonempty after reduction in \texttt{extract-tangles},
	as there are no lower regions.
	Furthermore, this region contains nontrivial bottom SCCs,
	since every vertex must have a successor in the region due to Lemma~\ref{lemma:alpha_maximal}.
	\qed
\end{proof}

\begin{lemma}
	\label{lemma:tangle_dominion}
	A tangle $t$ is a dominion if and only if $E_T(t)=\emptyset$
\end{lemma}
\begin{proof}
	If the tangle is a dominion, then player $\invalpha$ cannot leave it, therefore $E_T(t)=\emptyset$.
	If $E_T(t)=\emptyset$, then player $\invalpha$ cannot leave the tangle and since all plays consistent with $\sigma$ in the tangle are won by player $\alpha$, the tangle is a dominion.
	\qed
\end{proof}

\begin{lemma}
	\label{lemma:highest_dominion}
	$E_T(t)=\emptyset$ for every tangle $t$ found in the highest region of player $\alpha$.
\end{lemma}
\begin{proof}
	Per Lemma~\ref{lemma:E_t_higher}, $E_T(t)\subseteq\{\,v\in\dom(\r)\mid\r(v)\equiv_2 p\,\}$ when the tangle is found.
	There are no higher regions of player $\alpha$, therefore $E_T(t)=\emptyset$.
	\qed
\end{proof}

\begin{lemma}
	\label{lemma:search_dominion}
	The \texttt{search} algorithm terminates by finding a dominion.
\end{lemma}
\begin{proof}
	There is always a highest region of one of the players that is not empty.
	If a tangle is found in this region, then it is a dominion (Lemma~\ref{lemma:tangle_dominion}, Lemma~\ref{lemma:highest_dominion})
	and Algorithm~\ref{alg:search} terminates (line~9).
	If no tangle is found in this region, then the opponent can escape to a lower region.
	Thus, if no dominion is found in a highest region, then there is a lower region that contains a tangle (Lemma~\ref{lemma:lowest_tangle}) that must be unique (Lemma~\ref{lemma:unique_tangle}).
	As there are only finitely many unique tangles, eventually a dominion must be found.
	\qed
\end{proof}

\begin{lemma}
	\label{lemma:solve_solves}
	The \texttt{solve} algorithm solves parity games.
\end{lemma}
\begin{proof}
	Every invocation of \texttt{search} returns a dominion of the game (Lemma~\ref{lemma:search_dominion}).
	The $\alpha$-attractor of a dominion won by player $\alpha$ is also a dominion of player $\alpha$.
	Thus all vertices in $D$ are won by player $\alpha$.
	The winning strategy is derived as the witness strategy of $d$ with the strategy obtained by the attractor at line~6.
	At the end of \texttt{solve} every vertex of the game is either in $W_\Even$ or $W_\Odd$.
	\qed
\end{proof}

\subsection{Variations of tangle learning}

We propose three different variations of tangle learning that can be combined.

The first variation is \emph{alternating tangle learning}, where players take turns to maximally learn tangles, i.e., in a turn of player Even, we only search for tangles in regions of player Even, until no more tangles are found.
Then we search only for tangles in regions of player Odd, until no more tangles are found.
When changing players, the last decomposition can be reused.

The second variation is \emph{on-the-fly tangle learning}, where new tangles immediately refine the decomposition.
When new tangles are found, the decomposition procedure is reset to the highest region that attracts one of the new tangles, such that all regions in the top-down decomposition remain $\alpha$-maximal.
This is the region with priority $p := \max \{\; \min \{\; \r(v) \mid v\in E_T(t) \;\} \mid t\in A \;\}$.



A third variation skips the reduction step in \texttt{extract-tangles} and only extracts tangles from regions where none of the vertices of priority $p$ can escape to lower regions.
This still terminates finding a dominion, as Lemma~\ref{lemma:lowest_tangle} still applies, i.e., we always extract tangles from the lowest region.
Similar variations are also conceivable, such as only learning tangles from the lowest region.

\section{Complexity}
\label{sec:complexity}

We establish a relation between the time complexity of tangle learning and the number of 
 \emph{learned} tangles.

%


\begin{lemma}
\label{lemma:complex_attract}
Computing the top-down $\alpha$-maximal decomposition of a parity game runs in time $O(dm+dn\lvert T\rvert)$ given a parity game with $d$ priorities, $n$ vertices and $m$ edges, and a set of tangles $T$.
\end{lemma}
\begin{proof}
The attractor $\Attr^{\Game}_\alpha$ runs in time $O(n+m)$, if we record the number of remaining outgoing edges for each vertex~\cite{Verver2013}.
The attractor $\TAttr^{\Game,T}_\alpha$ runs in time $O(n+m+\lvert T\rvert+n\lvert T\rvert)$,
if implemented in a similar style.
As $m \geq n$, we simplify to $O(m+n\lvert T\rvert)$.
Since the decomposition computes at most $d$ regions, the decomposition runs in time $O(dm+dn\lvert T\rvert)$.
\qed
\end{proof}

\begin{lemma}
\label{lemma:complex_extract}
Searching for tangles in the decomposition runs in time $O(dm)$.
\end{lemma}
\begin{proof}
The \texttt{extract-tangles} procedure consists of a backward search, which runs in $O(n+m)$, and an SCC search based on Tarjan's algorithm, which also runs in $O(n+m)$.
This procedure is performed at most $d$ times (for each region).
As $m \geq n$, we simplify to $O(dm)$.
\qed
\end{proof}

\begin{lemma}
Tangle learning runs in time $O(dnm\lvert T\rvert+dn^2\lvert T\rvert^2)$ for a parity game with $d$ priorities, $n$ vertices, $m$ edges, and $\lvert T\rvert$ \emph{learned} tangles.
\end{lemma}
\begin{proof}
Given Lemma~\ref{lemma:complex_attract} and Lemma~\ref{lemma:complex_extract}, each iteration in \texttt{search} runs in time $O(dm+dn\lvert T\rvert)$.
The number of iterations is at most $\lvert T\rvert$, since we learn at least~$1$ tangle per iteration.
Then \texttt{search} runs in time $O(dm\lvert T\rvert+dn\lvert T\rvert^2)$.
Since each found dominion is then removed from the game, there are at most $n$ calls to \texttt{search}.
Thus tangle learning runs in time $O(dnm\lvert T\rvert+dn^2\lvert T\rvert^2)$.
\qed
\end{proof}

\begin{figure}[tb]
	\centering
	\scalebox{0.9}{
		\begin{tikzpicture}		
		\tikzset{every edge/.append style={>=stealth,->,solid,thick,draw,text height=0.5ex,text depth=0.2ex}}
		\tikzset{every node/.append style={minimum size=6mm,draw,fill=black!10}}
		\tikzset{my label/.style args={#1:#2}{
				append after command={
					($(\tikzlastnode.center)$) coordinate [label={[label distance=5mm,black]#1:\textbf{\strut #2}}]
		}}}
		\tikzstyle{even}=[diamond,minimum size=6mm,draw]
		\tikzstyle{odd}=[regular polygon,regular polygon sides=4,minimum size=6mm,draw]
		\tikzstyle{seven}=[even,fill=white]
		\tikzstyle{sodd}=[odd,fill=white]
		
		\draw (0.0,1.5)   node[seven, my label={above:a}] (a) {0};
		\draw (-1.5,1.5)  node[seven, my label={above:b}] (b) {2};
		\draw (-1.5,0)    node[seven, my label={below:c}] (c) {1};
		\draw (0.0,0)     node[seven, my label={below:d}] (d) {3};
		\draw (1.5,1.5)   node[sodd,  my label={above:e}] (e) {0};
		\draw (3.0,1.5)   node[sodd,  my label={above:f}] (f) {4};
		\draw (3.0,0.0)   node[sodd,  my label={below:g}] (g) {1};
		\draw (1.5,0.0)   node[sodd,  my label={below:h}] (h) {3};
		\draw (a) edge (b);
		\draw (a) edge [bend left=20] (e);
		\draw (e) edge [bend left=20] (a);
		\draw (b) edge (c);
		\draw (c) edge (d);
		\draw (d) edge (a);
		\draw (e) edge (f);
		\draw (f) edge (g);
		\draw (g) edge (h);
		\draw (h) edge (e);
		\draw (c) edge[in=-145,out=145,loop,looseness=7] (c);
		\draw (g) edge[in=-35,out=35,loop,looseness=4] (g);
		\end{tikzpicture}
	}
	\caption{A parity game that requires several turns to find a dominion.}
	\label{fig:4turns}
\end{figure}
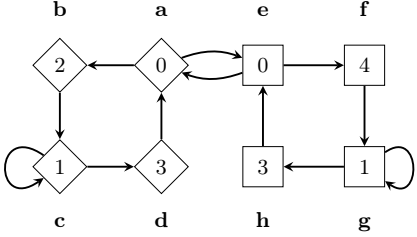

The complexity of tangle learning follows from the number of tangles that are learned before each dominion is found.
Often not all tangles in a game need to be learned to solve the game, only certain tangles.
Whether this number can be exponential in the worst case is an open question.
These tangles often serve to remove \emph{distractions} that prevent the other player from finding better tangles.
This concept is illustrated by the example in Figure~\ref{fig:4turns}, which requires multiple turns before a dominion is found.
The game contains 4 tangles:
$\{\,\textbf{c}\,\}$, $\{\,\textbf{g}\,\}$ (a dominion), $\{\,\textbf{a},\textbf{b},\textbf{c},\textbf{d}\,\}$ and $\{\,\textbf{a},\textbf{e}\,\}$.
The vertices $\{\,\textbf{e},\textbf{f},\textbf{g},\textbf{h}\,\}$ do not form a tangle, since the opponent wins the loop of vertex \textbf{g}. 
The tangle $\{\,\textbf{a},\textbf{b},\textbf{c},\textbf{d}\,\}$ is a dominion in the remaining game after $\Attr^\Game_\Odd(\{\,\textbf{g}\,\})$ has been removed.

The tangle $\{\,\textbf{g}\,\}$ is not found at first, as player Odd is distracted by \textbf{h}, i.e., prefers to play from \textbf{g} to \textbf{h}.
Thus vertex \textbf{h} must first be attracted by the opponent. 
This occurs when player Even learns the tangle $\{\,\textbf{a},\textbf{e}\,\}$, which is then attracted to \textbf{f}, which then attracts \textbf{h}.
However, the tangle $\{\,\textbf{a},\textbf{e}\,\}$ is blocked, as player Even is distracted by \textbf{b}.
Vertex \textbf{b} is attracted by player Odd when they learn the tangle $\{\,\textbf{c}\,\}$, which is attracted to \textbf{d}, which then attracts \textbf{b}.
So player Odd must learn tangle $\{\,\textbf{c}\,\}$ so player Even can learn tangle $\{\,\textbf{a},\textbf{e}\,\}$, which player Even must learn so player Odd can learn tangle $\{\,\textbf{g}\,\}$ and win the dominion $\{\,\textbf{e},\textbf{f},\textbf{g},\textbf{h}\,\}$, after which player Odd also learns $\{\,\textbf{a},\textbf{b},\textbf{c},\textbf{d}\,\}$ and wins the entire game.

One can also understand the algorithm as the players learning that their opponent can now play from some vertex $v$ via the learned tangle to a higher vertex $w$ that is won by the opponent.
In the example, we first learn that \textbf{b} actually leads to \textbf{d} via the learned tangle $\{\,\textbf{c}\,\}$.
Now \textbf{b} is no longer safe for player Even. 
However, player Even can now play from both \textbf{d} and \textbf{h} via the learned $0$-tangle $\{\,\textbf{a},\textbf{e}\,\}$ to \textbf{f}, so \textbf{d} and \textbf{h} are no longer interesting for player Odd and vertex \textbf{b} is again safe for player Even.

\section{Implementation}
\label{sec:implementation}

We implement four variations of tangle learning in the parity game solver Oink~\cite{Oink2018}.
Oink is a modern implementation of parity game algorithms written in C++.
Oink implements priority promotion~\cite{DBLP:conf/cav/BenerecettiDM16}, Zielonka's recursive algorithm~\cite{DBLP:journals/tcs/Zielonka98}, strategy improvement~\cite{DBLP:conf/cav/Fearnley17}, small progress measures~\cite{DBLP:conf/stacs/Jurdzinski00}, and quasi-polynomial time progress measures~\cite{DBLP:conf/spin/FearnleyJS0W17}.
Oink also implements self-loop solving and winner-controlled winning cycle detection, as proposed in~\cite{Verver2013}.
%
The implementation is publicly available via \url{https://www.github.com/trolando/oink}.

We implement the following variations of tangle learning: standard tangle learning (\texttt{tl}), alternating tangle learning (\texttt{atl}), on-the-fly tangle learning (\texttt{otftl}) and on-the-fly alternating tangle learning (\texttt{otfatl}).
The implementation mainly differs from the presented algorithm in the following ways.
We combine the \texttt{solve} and \texttt{search} algorithms in one loop.
We remember the highest region that attracts a new tangle and reset the decomposition to that region instead of recomputing the full decomposition.
In \texttt{extract-tangles}, we do not compute bottom SCCs for the highest region of a player, instead we return the entire reduced region as a single dominion (see also Lemma~\ref{lemma:highest_dominion}).



\section{Empirical evaluation}
\label{sec:evaluation}

The goal of the empirical evaluation is to study tangle learning and its variations on real-world examples and random games.
Due to space limitations, we do not report in detail on crafted benchmark families (generated by PGSolver~\cite{DBLP:conf/atva/FriedmannL09}), except that none of these games is difficult in runtime or number of tangles.

We use the parity game benchmarks from model checking and equivalence
checking proposed by Keiren~\cite{DBLP:conf/fsen/Keiren15} that are publicly available online.
These are $313$ model checking and $216$ equivalence checking games. 
We also consider random games, in part because the literature on parity games
tends to favor studying the behavior of algorithms on random games.
We include two classes of self-loop-free random games generated by PGSolver~\cite{DBLP:conf/atva/FriedmannL09} with a fixed number of vertices:
\begin{itemize}
	\item fully random games (\texttt{randomgame N N 1 N x}) \\
	$N \in \{\, 1000, 2000, 4000, 7000 \,\}$
	\item large low out-degree random games (\texttt{randomgame N N 1 2 x}) \\
	$N \in \{\, 10000, 20000, 40000, 70000, 100000, 200000, 400000, 700000, 1000000 \,\}$
\end{itemize}

We generate $20$ games for each parameter $N$, in total $80$ fully random games and $180$ low out-degree games.
All random games have $N$ vertices and up to $N$ distinct priorities.
We include low out-degree games,
since algorithms may behave differently on games where all vertices have few available moves,
as also suggested in~\cite{DBLP:conf/cav/BenerecettiDM16}.
In fact, as we see in the evaluation, fully random games appear trivial to solve, whereas games with few moves per vertex are more challenging.
Furthermore, the fully random games have fewer vertices but require more disk space (40~MB per compressed file for $N=7000$) than large low out-degree games (11~MB per compressed file for $N=1000000$).

We compare four variations of tangle learning to the implementations of Zielonka's recursive algorithm (optimized version of Oink) and of priority promotion (implemented in Oink by the authors of~\cite{DBLP:conf/cav/BenerecettiDM16}).
The motivation for this choice is that~\cite{Oink2018} shows that these are the fastest parity game solving algorithms.

In the following, we also use \emph{cactus plots} to compare the algorithms.
Cactus plots show that an algorithm solved $X$ input games within $Y$ seconds individually.

All experimental scripts and log files are available online via \url{https://www.
github.com/trolando/tl-experiments}.
The experiments were performed on a cluster of Dell PowerEdge M610 servers with two Xeon E5520 processors and 24~GB internal memory each.
The tools were compiled with gcc 5.4.0.

\subsection{Overall results}

\begin{table}[t]
	\centering
	\begin{tabu} to 0.7\linewidth {XXXX[0.8]X[0.8]}
		\toprule
		\textbf{Solver} & \textbf{MC\&EC} & \textbf{Random} & \multicolumn{2}{l}{\textbf{Random (large)}} \\
		&\small time&\small time&\small time&\small timeouts
		\\
		\midrule
%
  pp & 503 & 21 & 12770 & 6 \\ 
  zlk & 576 & 21 & 23119 & 13 \\ 
  otfatl & 808 & 21 & 2281 & 0 \\ 
  atl & 817 & 21 & 2404 & 0 \\ 
  otftl & 825 & 21 & 2238 & 0 \\ 
  tl & 825 & 21 & 2312 & 0 \\ 
  
		\bottomrule
	\end{tabu}
	\vspace{1em}
	\caption{Runtimes in sec. and number of timeouts (20 minutes) of the solvers Zielonka (\texttt{zlk}), priority promotion (\texttt{pp}), and tangle learning (\texttt{tl}, \texttt{atl}, \texttt{otftl}, \texttt{otfatl}).}
	\label{tbl:oink}
\end{table}

Table~\ref{tbl:oink} shows the cumulative runtimes of the six algorithms.
For the runs that timed out, we simply used the timeout value of $1200$ seconds, but this underestimates the actual runtime.

\subsection{Model checking and equivalence checking games}

\begin{figure}[t]
	\resizebox{\linewidth}{!}{
		\input{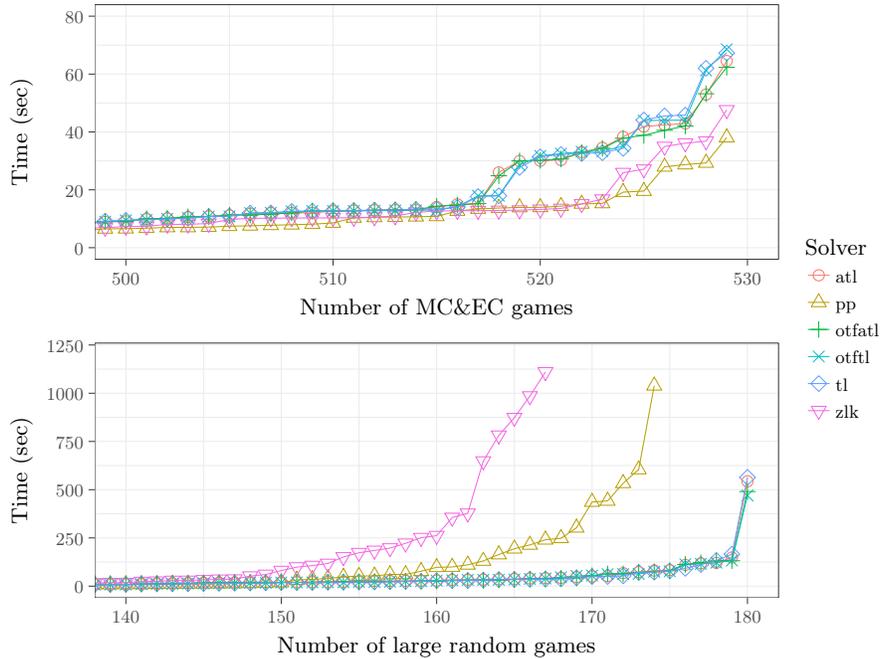}
	}
	\caption{Cactus plots of the solvers Zielonka (\texttt{zlk}), priority promotion (\texttt{pp}) and tangle learning (\texttt{tl}, \texttt{atl}, \texttt{otftl}, \texttt{otfatl}).
	The plot shows how many MC\&EC games (top) or large random games (bottom) are (individually) solved within the given time.
	}
	\label{fig:oink_rn2}
	\label{fig:oink_mceq}
\end{figure}

See Figure~\ref{fig:oink_mceq} for the cactus plot of the six solvers on model checking and equivalence checking games.
This graph suggests that for most games, tangle learning is only slightly slower than the other algorithms.
The tangle learning algorithms require at most $2\times$ as much time for $12$ of the $529$ games.
There is no significant difference between the four variations of tangle learning.

The $529$ games have on average $1.86$ million vertices and $5.85$ million edges, and at most $40.6$ million vertices and $167.5$ million edges.
All equivalence checking games have $2$ priorities,
so every tangle is a dominion. 
The model checking games have $2$ to $4$ priorities.
Tangle learning learns non-dominion tangles for only $30$ games, and more than $1$ tangle only for the $22$ games that check the \verb|infinitely_often_read_write| property.
The most extreme case is 1,572,864 tangles for a game with 19,550,209 vertices.
These are all $0$-tangles that are then attracted to become part of $2$-dominions.

That priority promotion and Zielonka's algorithm perform well is no surprise.
See also Sec.~\ref{sec:tangles_in_pp}.
Solving these parity games requires few iterations for all algorithms, but tangle learning spends more time learning and attracting individual tangles, which the other algorithms do not do.
Zielonka requires at most $27$ iterations, priority promotion at most $28$ queries and $9$ promotions.
%
Alternating tangle learning requires at most $2$ turns.
We conclude that these games are not complex and that their difficulty is related to their sheer size.





\subsection{Random games}

Table~\ref{tbl:oink} shows no differences between the algorithms for the fully random games.
Tangle learning learns no tangles except dominions for any of these games.
This agrees with the intuition that the vast number of edges in these games lets attractor-based algorithms quickly attract large portions of the game.


See Fig.~\ref{fig:oink_rn2} for a cactus plot of the solvers on the larger random games.
Only $167$ games were solved within $20$ minutes each by Zielonka's algorithm and only $174$ games by priority promotion.
See Table~\ref{tbl:hardrandom} for details of the slowest $10$ random games for alternating tangle learning.
There is a clear correlation between the runtime, the number of tangles and the number of turns. 
One game is particularly interesting, as it requires significantly more time than the other games.

The presence of one game that is much more difficult is a feature of using random games.
It is likely that if we generated a new set of random games, we would obtain different results.
This could be ameliorated by experimenting on thousands of random games and even then it is still a game of chance whether some of these random games are significantly more difficult than the others.

\begin{table}[tb]
	\centering
	\begin{tabu} to \linewidth {X[16]X[12r]X[12r]X[12r]X[12r]X[10r]X[10r]X[10r]X[10r]X[10r]X[10r]}
		\toprule
		\textbf{Time} & 543 & 148 & 121 & 118 & 110 & 83 & 81 & 73 & 68 & 52 \\
		\textbf{Tangles} & 4,018 & 1,219 & 737 & 560 & 939 & 337 & 493 & 309 & 229 & 384 \\
		\textbf{Turns} & 91 & 56 & 23 & 25 & 30 & 12 & 18 & 10 & 10 & 18 \\
		\textbf{Size} & 1M & 1M & 700K & 1M & 700K & 1M & 1M & 1M & 1M & 1M \\
		\bottomrule
	\end{tabu}
	\vspace{1em}
	\caption{The $10$ hardest random games for the \texttt{atl} algorithm, with time in seconds and size in number of vertices.}
	\label{tbl:hardrandom}
\end{table}

\section{Tangles in other algorithms}
\label{sec:related}

We argue that tangles play a fundamental role in various other parity game solving algorithms.
We refer to~\cite{Oink2018} for descriptions of these algorithms.

\subsection{Small progress measures}

The small progress measures algorithm~\cite{DBLP:conf/stacs/Jurdzinski00} iteratively performs local updates to vertices
until a fixed point is reached.
Each vertex is equipped with some measure that records a statistic of the best game either player knows that they can play from that vertex so far.
By updating measures based on the successors, they essentially play the game backwards.
When they can no longer perform updates, the final measures indicate the
winning player of each vertex.

The measures in small progress measures record how often each even priority is encountered along the most optimal play (so far) until a higher priority is encountered.
As argued in~\cite{Oink2018} and~\cite{DBLP:journals/corr/GazdaW15},
player Even tries to visit vertices with even priorities as often as possible, while prioritizing plays with more higher even priorities.
This often resets progress for lower priorities.
Player Odd has the opposite goal, i.e., player Odd prefers to play to a lower even priority to avoid a higher even priority, even if the lower priority is visited infinitely often.
When the measures record a play from some vertex that visits more vertices with some even priority than exist in the game, this indicates that player Even can force player Odd into a cycle, unless they concede and play to a higher even priority.
A mechanism called cap-and-carryover~\cite{Oink2018} ensures via slowly rising measures that the opponent is forced to play to a higher even priority.

We argue that when small progress measures finds cycles of some priority $p$,
this is due to the presence of a $p$-tangle, namely precisely those vertices whose measures increase beyond the number of vertices with priority $p$, since these measures can only increase so far in the presence of cycles out of which the opponent cannot escape except by playing to vertices with a higher even priority.

One can now understand small progress measures as follows.
The algorithm indirectly searches for tangles of player Even, and then searches for the best escape for player Odd by playing to the lowest higher even priority.
If no such escape exists for a tangle, then the measures eventually rise to $\top$, indicating that player Even has a dominion.
Whereas tangle learning is affected by \emph{distractions}, small progress measures is driven by the dual notion of \emph{aversions}, i.e., high even vertices that player Odd initially tries to avoid.
The small progress measures algorithm tends to find tangles repeatedly, especially when they are nested.


\subsection{Quasi-polynomial time progress measures}

The quasi-polynomial time progress measures algorithm~\cite{DBLP:conf/spin/FearnleyJS0W17} is similar to small progress measures. 
This algorithm records the number of dominating even vertices along a play, i.e., such that every two such vertices are higher than all intermediate vertices.
For example, in the path $1\underline{2}131\underline{4}23\underline{2}15\underline{6}3\underline{2}1\underline{2}$, all vertices are dominated by each pair of underlined vertices of even priority.
Higher even vertices are preferred, even if this (partially) resets progress on lower priorities.

Tangles play a similar role as with small progress measures.
The presence of a tangle lets the value iteration procedure increase the measure up to the point where the other player ``escapes'' the tangle via a vertex outside of the tangle. 
This algorithm has a similar weakness to nested tangles, but it is less severe as progress on lower priorities is often retained.
In fact, the lower bound game in~\cite{DBLP:conf/spin/FearnleyJS0W17}, for which the quasi-polynomial time algorithm is slow, is precisely based on nested tangles and is easily solved by tangle learning.

%
%
%

\subsection{Strategy improvement}

As argued by Fearnley~\cite{DBLP:conf/lpar/Fearnley10}, tangles play a fundamental role in the behavior of strategy improvement.
Fearnley writes that instead of viewing strategy improvement as a process that tries to increase valuations, one can view it as a process that tries to force ``consistency with snares''~\cite[Sec.~6]{DBLP:conf/lpar/Fearnley10}, i.e., 
as a process that searches for escapes from tangles.

\subsection{Priority promotion}
\label{sec:tangles_in_pp}

Priority promotion~\cite{DBLP:conf/cav/BenerecettiDM16,Benerecetti2018} computes a top-down $\alpha$-maximal decomposition and identifies ``closed $\alpha$-regions'', i.e., regions where the losing player cannot escape to lower regions.
A closed $\alpha$-region is essentially a collection of possibly unconnected tangles and vertices that are attracted to these tangles.
Priority promotion then promotes the closed region to the lowest higher region that the losing player can play to, i.e., the lowest region that would attract one of the tangles in the region.
Promoting is different from attracting, as tangles in a region can be promoted to a priority that they are not attracted to.
Furthermore, priority promotion has no mechanism to remember tangles, so the same tangle can be discovered many times.
This is somewhat ameliorated in extensions such as region recovery~\cite{DBLP:conf/hvc/BenerecettiDM16} and delayed promotion~\cite{DBLP:journals/corr/BenerecettiDM16}, which aim to decrease how often regions are recomputed.

Priority promotion has a good practical performance for games where computing and attracting individual tangles is not necessary, e.g., when tangles are only attracted once and all tangles in a closed region are attracted to the same higher region, as is the case with the benchmark games of~\cite{DBLP:conf/fsen/Keiren15}.

\subsection{Zielonka's recursive algorithm}

Zielonka's recursive algorithm~\cite{DBLP:journals/tcs/Zielonka98} also computes a top-down $\alpha$-maximal decomposition, but instead of attracting from lower regions to higher regions, the algorithm attracts from higher regions to tangles in the subgame.
Essentially, the algorithm starts with the tangles in the lowest region and attracts from higher regions to these tangles.
When vertices from a higher $\alpha$-region are attracted to tangles of player $\invalpha$, progress for player $\alpha$ is reset.
Zielonka's algorithm also has no mechanism to store tangles and games that are exponential for Zielonka's algorithm, such as in~\cite{DBLP:conf/gandalf/BenerecettiDM17}, are trivially solved by tangle learning.

\section{Conclusions}
\label{sec:conclusions}

We introduced the notion of a tangle as a subgraph of the game where one player knows how to win all cycles.
We showed how tangles and nested tangles play a fundamental role in various parity game algorithms.
These algorithms are not explicitly aware of tangles and can thus repeatedly explore the same tangles. 
We proposed a novel algorithm called tangle learning, which identifies tangles in a parity game and then uses these tangles to attract sets of vertices at once.
The key insight is that tangles can be used with the attractor to form bigger (nested) tangles and, eventually, dominions.
We evaluated tangle learning in a comparison with priority promotion and Zielonka's recursive algorithm and showed that tangle learning is competitive for model checking and equivalence checking games, and outperforms other solvers for large random games.


We repeat Fearnley's assertion~\cite{DBLP:conf/lpar/Fearnley10} that ``a thorough and complete understanding of how snares arise in a game is a necessary condition for devising a polynomial time algorithm for these games.''
Fearnley also formulated the challenge to give a clear formulation of how the structure of tangles in a given game affects the difficulty of solving it.
We propose that a difficult game for tangle learning must be one that causes alternating tangle learning to have many turns before a dominion is found.

\section*{Acknowledgements}

We thank the anonymous referees for their helpful comments, Jaco
van de Pol for the use of the computer cluster, and Armin Biere
for generously supporting this research.

\clearpage

\bibliographystyle{splncs03}
\bibliography{lit}

\begin{thebibliography}{10}
\providecommand{\url}[1]{\texttt{#1}}
\providecommand{\urlprefix}{URL }

\bibitem{DBLP:journals/corr/BenerecettiDM16}
Benerecetti, M., Dell'Erba, D., Mogavero, F.: {A Delayed Promotion Policy for
  Parity Games}. In: {GandALF} 2016. {EPTCS}, vol. 226, pp. 30--45 (2016)

\bibitem{DBLP:conf/hvc/BenerecettiDM16}
Benerecetti, M., Dell'Erba, D., Mogavero, F.: {Improving Priority Promotion for
  Parity Games}. In: {HVC} 2016. LNCS, vol. 10028, pp. 117--133 (2016)

\bibitem{DBLP:conf/cav/BenerecettiDM16}
Benerecetti, M., Dell'Erba, D., Mogavero, F.: {Solving Parity Games via
  Priority Promotion}. In: {CAV} 2016. LNCS, vol. 9780, pp. 270--290. Springer
  (2016)

\bibitem{DBLP:conf/gandalf/BenerecettiDM17}
Benerecetti, M., Dell'Erba, D., Mogavero, F.: Robust exponential worst cases
  for divide-et-impera algorithms for parity games. In: {GandALF}. {EPTCS},
  vol. 256, pp. 121--135 (2017)

\bibitem{Benerecetti2018}
Benerecetti, M., Dell'Erba, D., Mogavero, F.: Solving parity games via priority
  promotion. Formal Methods in System Design  (2018)

\bibitem{DBLP:conf/stoc/CaludeJKL017}
Calude, C.S., Jain, S., Khoussainov, B., Li, W., Stephan, F.: Deciding parity
  games in quasipolynomial time. In: {STOC}. pp. 252--263. {ACM} (2017)

\bibitem{Oink2018}
van Dijk, T.: Oink: an implementation and evaluation of modern parity game
  solvers. In: {TACAS} (2018), \url{https://arxiv.org/pdf/1801.03859}

\bibitem{DBLP:conf/focs/EmersonJ91}
Emerson, E.A., Jutla, C.S.: Tree automata, mu-calculus and determinacy
  (extended abstract). In: {FOCS}. pp. 368--377. {IEEE} Computer Society (1991)

\bibitem{DBLP:journals/tcs/EmersonJS01}
Emerson, E.A., Jutla, C.S., Sistla, A.P.: On model checking for the mu-calculus
  and its fragments. Theor. Comput. Sci.  258(1-2),  491--522 (2001)

\bibitem{DBLP:conf/lpar/Fearnley10}
Fearnley, J.: Non-oblivious strategy improvement. In: {LPAR} (Dakar). LNCS,
  vol. 6355, pp. 212--230. Springer (2010)

\bibitem{DBLP:conf/cav/Fearnley17}
Fearnley, J.: Efficient parallel strategy improvement for parity games. In:
  {CAV} {(2)}. LNCS, vol. 10427, pp. 137--154. Springer (2017)

\bibitem{DBLP:conf/spin/FearnleyJS0W17}
Fearnley, J., Jain, S., Schewe, S., Stephan, F., Wojtczak, D.: An ordered
  approach to solving parity games in quasi polynomial time and quasi linear
  space. In: {SPIN}. pp. 112--121. {ACM} (2017)

\bibitem{DBLP:conf/atva/FriedmannL09}
Friedmann, O., Lange, M.: Solving parity games in practice. In: {ATVA}. {LNCS},
  vol. 5799, pp. 182--196. Springer (2009)

\bibitem{DBLP:journals/corr/GazdaW15}
Gazda, M., Willemse, T.A.C.: Improvement in small progress measures. In:
  GandALF. {EPTCS}, vol. 193, pp. 158--171 (2015)

\bibitem{DBLP:conf/dagstuhl/2001automata}
Gr{\"{a}}del, E., Thomas, W., Wilke, T. (eds.): Automata, Logics, and Infinite
  Games: {A} Guide to Current Research [outcome of a Dagstuhl seminar, February
  2001], LNCS, vol. 2500. Springer (2002)

\bibitem{DBLP:journals/ipl/Jurdzinski98}
Jurdzinski, M.: Deciding the winner in parity games is in {UP} $\cap$ co-{UP}.
  Inf. Process. Lett.  68(3),  119--124 (1998)

\bibitem{DBLP:conf/stacs/Jurdzinski00}
Jurdzinski, M.: Small progress measures for solving parity games. In: {STACS}.
  LNCS, vol. 1770, pp. 290--301. Springer (2000)

\bibitem{DBLP:journals/siamcomp/JurdzinskiPZ08}
Jurdzinski, M., Paterson, M., Zwick, U.: A deterministic subexponential
  algorithm for solving parity games. {SIAM} J. Comput.  38(4),  1519--1532
  (2008)

\bibitem{DBLP:conf/fsen/Keiren15}
Keiren, J.J.A.: Benchmarks for parity games. In: {FSEN}. LNCS, vol. 9392, pp.
  127--142. Springer (2015)

\bibitem{DBLP:journals/tcs/Kozen83}
Kozen, D.: Results on the propositional mu-calculus. Theor. Comput. Sci.  27,
  333--354 (1983)

\bibitem{DBLP:conf/stoc/KupfermanV98}
Kupferman, O., Vardi, M.Y.: Weak alternating automata and tree automata
  emptiness. In: {STOC}. pp. 224--233. {ACM} (1998)

\bibitem{DBLP:journals/jcss/Schewe17}
Schewe, S.: Solving parity games in big steps. J. Comput. Syst. Sci.  84,
  243--262 (2017)

\bibitem{Verver2013}
Verver, M.: Practical improvements to parity game solving. Master's thesis,
  University of Twente (2013)

\bibitem{DBLP:conf/cav/VogeJ00}
V{\"{o}}ge, J., Jurdzinski, M.: A discrete strategy improvement algorithm for
  solving parity games. In: {CAV}. LNCS, vol. 1855, pp. 202--215. Springer
  (2000)

\bibitem{DBLP:journals/tcs/Zielonka98}
Zielonka, W.: Infinite games on finitely coloured graphs with applications to
  automata on infinite trees. Theor. Comput. Sci.  200(1-2),  135--183 (1998)

\end{thebibliography}

\clearpage
\appendix

\end{document}